\documentclass[12pt]{article}

\usepackage{bm}
\usepackage{amssymb}
\usepackage{amsmath}
\usepackage{amsthm}
\usepackage{dsfont}
\usepackage{thmtools} 
\usepackage{thm-restate}
\usepackage{booktabs}
\usepackage{multirow}
\usepackage{adjustbox}
\usepackage{xcolor}
\usepackage{natbib}
\usepackage{hyperref}
\usepackage{url}
\usepackage{appendix}
\usepackage[ruled,vlined]{algorithm2e}

\def\mathP{\mathrm{Pr}}

\def\mathR{\mathbb{R}}

\newtheorem{theorem}{Theorem}[section]

\newtheorem{lemma}[theorem]{Lemma}
\newtheorem{proposition}[theorem]{Proposition}

\newcommand{\blind}{1}

\begin{document}

\def\spacingset#1{\renewcommand{\baselinestretch}%
{#1}\small\normalsize} \spacingset{1}

\if1\blind
{
	\title{Modeling Spatial Extremes Using Normal Mean-Variance Mixtures}
	\author{Zhongwei Zhang\\
	CEMSE Division, King Abdullah University of Science and Technology\\
	and \\
	Rapha\"el Huser \\
	CEMSE Division, King Abdullah University of Science and Technology\\
	and \\
	Thomas Opitz\\
	Biostatistics and Spatial Processes, INRAE\\
	and\\
	Jennifer Wadsworth\\
	Department of Mathematics and Statistics, Lancaster University}
	\date{}
	\maketitle
} \fi

\if0\blind
{
	\bigskip
	\bigskip
	\bigskip
	\begin{center}
		{\LARGE\bf Title}
	\end{center}
	\medskip
} \fi

\medskip

\begin{abstract}
    Classical models for multivariate or spatial extremes are mainly based upon the asymptotically justified max-stable or generalized Pareto processes. These models are suitable when asymptotic dependence is present, i.e., the joint tail decays at the same rate as the marginal tail. However, recent environmental data applications suggest that asymptotic independence is equally important and, unfortunately, existing spatial models in this setting that are both flexible and can be fitted efficiently are scarce. Here, we propose a new spatial copula model based on the generalized hyperbolic distribution, which is a specific normal mean-variance mixture and is very popular in financial modeling. The tail properties of this distribution have been studied in the literature, but with contradictory results. It turns out that the proofs from the literature contain mistakes. We here give a corrected theoretical description of its tail dependence structure and then exploit the model to analyze a simulated dataset from the inverted Brown--Resnick process, hindcast significant wave height data in the North Sea, and wind gust data in the state of Oklahoma, USA. We demonstrate that our proposed model is flexible enough to capture the dependence structure not only in the tail but also in the bulk.
\end{abstract}

\spacingset{1.4} 

\section{Introduction}
The statistical modeling of spatial extremes has gained significant interest in recent decades due to the increasing occurrence and sizes of natural extreme events, such as heat waves, heavy rainfall, and wildfires.
When modeling spatial extremes, accurate inference for the marginal distribution at each site and a precise assessment of the dependence structure of extreme events among different sites are both needed; see \cite{Davison2012,Davison2015} and \cite{HuserW2020} for an overview.
In this article we focus on modeling the spatial dependence of extreme events.

Classical models for spatial extremes are mainly based upon the asymptotically justified max-stable processes \citep{deHaan1984,deHaan2006} or generalized Pareto processes \citep{Rootzen2006,Ferreira2014}. 
The limiting dependence structures that arise in these models must be either asymptotically dependent (defined in Section \ref{taildependence}), i.e., the joint tail decays at the same rate as the marginal tail, or for maxima, exactly independent.
However, asymptotic independence appears to be equally important, as suggested by recent environmental data applications \citep{Wadsworth2012,Le2018}.
Unfortunately, existing models in this setting that are both flexible and can be fitted efficiently are scarce.

Existing models from the literature that can capture asymptotic independence include the Gaussian copula model \citep{Bortot2000}, inverted max-stable models \citep{Wadsworth2012}, the Huser--Wadsworth model \citep{Huser2019}, and the conditional extremes model \citep{WadsworthT2019}.
The Gaussian copula is the simplest model and has a restrictive dependence structure; inverted max-stable models present the same difficulties for inference as max-stable models since often only likelihoods based on lower-dimensional densities are available \citep{Padoan2010,Castruccio2016}; the Huser--Wadsworth model can capture both asymptotic dependence and independence but its distribution and (potentially censored) density functions rely on unidimensional integrals which have to be computed numerically, and this is computationally prohibitive in high dimensions, especially when computation of the multivariate normal distribution function is required; the conditional extremes model allows the change of asymptotic dependence class with distance between sites but it lacks an unconditional interpretation; see \cite{HuserW2020} for a more detailed discussion.

Here we propose a flexible model based on an extension of scale mixtures.
A random vector $\bm{X}\in\mathR^d$ is called a scale mixture if it has the representation $\bm{X}=R \bm{W}$, where $R$ is a non-negative univariate random variable and $\bm{W}\in\mathR^d$ is a random vector. 
Scale mixtures provide a flexible family of distributions that can capture both asymptotic dependence and independence depending on the specification of the tail of $R$ and $\bm{W}$; see \cite{HOT}, \cite{Huser2019}, \cite{EOW2019}.
In practice $\bm{W}$ is often taken as a Gaussian random vector and in this case $\bm{X}$ is termed a Gaussian scale mixture.
Natural alternatives and extensions of Gaussian scale mixtures include the Gaussian location mixtures \citep{KHG}, and the skew-$t$ model \citep{Morris2017}, which is a specific Gaussian location-scale mixture.
Here we focus on a different form of Gaussian location-scale mixtures, for which the term \emph{normal mean-variance mixtures} has been coined in the literature. 
Specifically, $\bm{X}$ is called a normal mean-variance mixture if it can be represented as
\begin{equation}\label{NMVM_repre}
\bm{X} = \bm{\mu} + \bm{\gamma}R + \sqrt{R}\bm{W}, \quad R \perp \!\!\! \perp  \bm{W},
\end{equation}
where $\bm{\mu}\in\mathR^d$ is a location parameter vector, $\bm{\gamma}\in\mathR^d$ regulates the skewness, $R\sim F_R$ is a non-negative mixing random variable, and $\bm{W}\sim\mathcal{N}_d(\bm{0},\Sigma)$ is a Gaussian random vector with correlation matrix $\Sigma$.
This provides a very flexible family of distributions and by properly choosing the parameters and mixing distribution, many well-known multivariate distributions can be obtained.
Another benefit of normal mean-variance mixtures is the simple implementation of their conditional simulation, using their property of closedness under conditioning \citep{Jamalizadeh2019}, which is in clear contrast to max-stable or inverted max-stable models.

A prominent example of normal mean-variance mixtures is the generalized hyperbolic (GH) distribution, when the mixing distribution $F_R$ is taken as generalized inverse Gaussian.
The skew-$t$ model \citep{Morris2017} can be represented as $\bm{X} = \bm{\mu} + \gamma_c\bm{1}_d\sqrt{R}|Z| + \sqrt{R}\bm{W}, \gamma_c\in\mathR$, where $\bm{1}_d$ is a $d$-dimensional vector of $1$s, $R$ has an inverse Gamma distribution, which is a limiting case of the generalized inverse Gaussian distribution, $Z$ is a standard normal random variable, $\bm{W}$ remains a Gaussian random vector and $R, Z, \bm{W}$ are mutually independent.
Although the representations of this skew-$t$ model and the GH distribution look similar, a major difference between them is the extra random variable $Z$ in the representation of the skew-$t$ model, and their tail dependence structures are significantly different.
That is, the skew-$t$ model is asymptotically dependent, and the GH distribution has more free parameters and is asymptotically independent except in one limiting case; see Section \ref{section:GHtail} for more details.

The GH distribution is very popular in financial modeling \citep{BarndorffNielsen1997,Prause1999} thanks to its flexible univariate distribution and infinite divisibility; more details of this distribution are given in Section \ref{GH_intro}.
Due to the popularity of this distribution in finance, its tail properties have been studied in the literature \citep{Nolde14,Hammerstein16}, but surprisingly with contradictory results.
By examining their proofs in detail, we have found that both of them contain subtle mistakes.
Here, we point out the mistakes in their proofs, which lead to the contradiction, and give a corrected description of the tail dependence of the GH distribution.

Based on this result, we propose to use the GH copula for spatial extremes, which has been less investigated compared with its flexible univariate distribution, and to our knowledge not yet exploited in environmental applications.
This model provides flexible dependence structures with many subclasses and limiting models, such as the normal inverse Gaussian, hyperbolic,  Gaussian and student $t$ copulas. 
We apply the model using a full likelihood approach (thus, avoiding the computationally prohibitive censoring mechanism) to a simulated dataset from the inverted Brown--Resnick process, the hindcast significant wave height data considered in \cite{Wadsworth2012} and \cite{Huser2019}, and wind gust data in the state of Oklahoma, USA, and demonstrate that our proposed model is flexible enough to capture the dependence structure not only in the tail but also in the bulk.

This paper is structured as follows. In Section 2 we introduce the multivariate generalized hyperbolic distribution and measures of tail dependence, and we then review existing results on the tail dependence properties of the GH distribution and present a corrected description.
Section 3 introduces the copula-based likelihood inference, followed by two simulation studies, and Section 4 consists of two data applications. 
Section 5 concludes with a discussion.

\section{Modeling}
\subsection{The Generalized Hyperbolic (GH) Distribution}\label{GH_intro}
The GH distribution is a specific normal mean-variance mixture with representation (\ref{NMVM_repre}) and mixing distribution $F_R$ specified to be generalized inverse Gaussian, denoted as GIG$(\lambda,\kappa,\psi)$. 
The GIG$(\lambda,\kappa,\psi)$ probability density function is
\[
f_\text{GIG}(x) = \Big(\frac{\psi}{\kappa}\Big)^{\lambda/2}\frac{x^{\lambda-1}}{2K_\lambda(\sqrt{\kappa\psi})}\exp\Big\{-\frac{1}{2}\Big(\frac{\kappa}{x}+\psi x\Big) \Big\}, \quad x>0,
\]
where $K_\lambda$ is the modified Bessel function of the second kind with index $\lambda$, and, for $\bm{x}\in\mathR^d$, the probability density function of a $d$-dimensional GH distribution is
\[
    f_\text{GH}(\bm{x}) = a_d \cdot \frac{K_{\lambda-d/2}\Big[\sqrt{\{\kappa+(\bm{x}-\bm{\mu})^\top \Sigma^{-1}(\bm{x}-\bm{\mu})\}(\psi+\bm{\gamma}^\top \Sigma^{-1}\bm{\gamma})}\Big] e^{(\bm{x}-\bm{\mu})^\top\Sigma^{-1}\bm{\gamma}}} {\Big[\sqrt{\{\kappa+(\bm{x}-\bm{\mu})^\top \Sigma^{-1}(\bm{x}-\bm{\mu})\}(\psi+\bm{\gamma}^\top \Sigma^{-1}\bm{\gamma})} \Big]^{d/2 - \lambda}},
\]
where 
\[
a_d = \frac{\psi^{\lambda/2} (\psi+\bm{\gamma}^\top \Sigma^{-1}\bm{\gamma})^{d/2 - \lambda}}{(2\pi)^{d/2}|\Sigma|^{1/2}\kappa^{\lambda/2}K_\lambda(\sqrt{\kappa \psi})},
\]
$\bm{\mu}\in\mathR^d$ is a location parameter vector, $\bm{\gamma}\in\mathR^d$ regulates the skewness, $\Sigma\in\mathR^{d\times d}$ is a positive definite dispersion matrix, and $\lambda, \kappa, \psi \in\mathR$ control the shape of the mixing GIG distribution.
We denote this GH distribution by $\bm{X} \sim \text{GH}_d(\lambda,\kappa,\psi,\bm{\gamma},\bm{\mu},\Sigma)$.
Both in the GIG and GH distributions, the admissible parameter values are $\lambda<0, \kappa>0, \psi\geq 0$ or $\lambda=0, \kappa>0, \psi>0$ or $\lambda>0, \kappa\geq 0, \psi>0$.
Note that $\psi=0$ and $\kappa=0$ should be understood as two limiting cases.

One important property of the GH distribution is that it is closed under marginalization, conditioning and linear transformations. 
Specifically, if $\bm{X} \sim \text{GH}_d(\lambda,\kappa,\psi,\bm{\gamma},\bm{\mu},\Sigma)$ and $\bm{Y}=B\bm{X}+\bm{b}$, where $B\in\mathR^{k\times d}$ and $\bm{b}\in\mathR^k$, then $\bm{Y}\sim \text{GH}_k(\lambda,\kappa,\psi,B\bm{\gamma},B\bm{\mu}+\bm{b},B\Sigma B^\top)$.
This distribution is a rich family with many special subclasses and limiting cases; see Table \ref{tab:GHsubclass} for some examples and the corresponding admissible parameter domains.
Note that, for the GH, hyperbolic, and NIG distributions the dispersion matrix $\Sigma$ needs to be a correlation matrix for identifiability reasons.
This distribution was first introduced in \cite{Barndorff1977} to model sand sizes and has become very popular in modeling turbulence and returns of financial assets; see \cite{BarndorffNielsen1997}, \cite{Prause1999}, and \cite{Schimidt2003}.
It is also worthwhile to mention that the GH distribution is elliptical if and only if $\bm{\gamma}=\bm{0}$; see Corollary 3 in \cite{Hammerstein16}.
More details on its properties can be found in \cite{BJ81}, \cite{Prause1999} and \cite{McNeil2005}. 

\begin{table}
\centering
\caption{Some special subclasses and limiting cases of the GH distribution}
\begin{tabular}{lcccc}\toprule
{} &\multicolumn{4}{c}{Parameter domain}  \\
Distribution & $\lambda$ & $\kappa$ & $\psi$ & $\bm{\gamma}$ \\\midrule
Hyperbolic & $(d+1)/2$ & $>0$ & $>0$ & $\in\mathR^d$ \\
Normal inverse Gaussian (NIG) & $-1/2$ & $>0$ & $>0$ & $\in\mathR^d$  \\
Student-$t$ with {\rm df} degrees of freedom & $-{\rm df}/2<0$ & ${\rm df}>0$ & $0$ & $\bm{0}$ \\
Cauchy & $-1/2$ & $1$ & $0$ & $\bm{0}$ \\
\bottomrule
\end{tabular}
\label{tab:GHsubclass}
\end{table}

There are several equivalent parametrizations for the GH distribution.
The most common one is to parametrize it as GH$(\lambda,\alpha,\bm{\mu},\Delta,\delta,\bm{\beta})$,
where we let $\Delta=|\Sigma|^{-1/d}\Sigma$, $\bm{\beta}=\Sigma^{-1}\bm{\gamma}$, $\delta=\sqrt{|\Sigma|^{1/d}\kappa}$, and $\alpha=\sqrt{|\Sigma|^{-1/d}(\psi+\bm{\gamma}^\top\Sigma^{-1}\bm{\gamma})}$. 
This parametrization was used in \cite{Barndorff1977}, but it does not give the nice property that the important parameters $\alpha$ and $\delta$ are invariant under linear transformations.
Therefore, we here instead adopt the parametrization $\text{GH}_d(\lambda,\kappa,\psi,\bm{\gamma},\bm{\mu},\Sigma)$.

\subsection{Measures of Tail Dependence} \label{taildependence}
\cite{CHT1999} proposed the tail dependence coefficient $\chi$, providing it exists, to measure the extremal dependence between random variables $X_1$ and $X_2$.
It may be defined through the limit
\begin{equation}\label{chidefinition}
\chi_u := \frac{\mathP\{F_1 (X_1) >u, F_2 (X_2) >u\}}{\mathP\{F_1(X_1) > u \}} \rightarrow \chi \text{ as } u \rightarrow 1,
\end{equation}
where $F_1$ and $F_2$ are the marginal distribution functions of $X_1$ and $X_2$, respectively, and are assumed to be continuous without loss of generality.
The lower tail dependence coefficient can be obtained via reflection.
Throughout this work we focus on the upper tail dependence.
The vector $(X_1, X_2)^\top$ is termed asymptotically independent if $\chi = 0$ and asymptotically dependent if $\chi > 0$. 

When asymptotic independence is present, one useful quantity to measure the residual dependence at pre-asymptotic levels is the residual tail dependence coefficient $\eta$ \citep{Ledford1996}, which is defined by the following asymptotic expansion of the joint tail (providing it exists):
\begin{equation}\label{etadefinition}
\mathP\{F_1 (X_1) >u, F_2 (X_2) >u\} \sim \mathcal{L}\{(1-u)^{-1}\}(1-u)^{1/\eta} \text{ as } u \rightarrow 1,
\end{equation}
where $\mathcal{L}$ is a slowly varying function, i.e., $\mathcal{L}(tx)/\mathcal{L}(x)\rightarrow 1$ as $x\rightarrow\infty$ for all $t>0$, and $f(x)\sim g(x)$ as $x\rightarrow x_0$ means that, for $g$ non-zero in a neighbourhood of $x_0$, $f(x)/g(x) \rightarrow 1$ as $x\rightarrow x_0$.
The coefficient $\eta\in(0,1]$ determines the joint tail decay rate.
A reformulation of (\ref{etadefinition}) implies that $\eta$ can be alternatively defined through the limit
\begin{equation}\label{etadefinition2}
\eta_u := \frac{\log \mathP\{F_1(X_1) > u \} }{\log \mathP\{F_1(X_1) > u, F_2(X_2) > u\}} \rightarrow \eta \text{ as } u \rightarrow 1,
\end{equation}
and this definition allows one to estimate $\eta$ with $\eta_u$ for $u$ close to $1$.

The relationship between these two above coefficients is $\eta<1 \Longrightarrow \chi = 0$ and $\chi>0 \Longrightarrow \eta = 1$.
The converse of this relationship does not necessarily hold; see Proposition 4 in \cite{Manner2011} and an example with $\eta=1, \chi=0$ in \cite{Huser2019}.
We say that $(X_1, X_2)^\top$ is (a) positively associated if $1/2<\eta\leq 1$; (b) near-independent if $\eta=1/2$; (c) negatively associated if $0<\eta<1/2$.
The bivariate measures $\chi$ and $\eta$ can be easily extended to dimension $d>2$ by replacing $\mathP\{F_1 (X_1) >u, F_2 (X_2) >u\}$ in (\ref{chidefinition}) and (\ref{etadefinition2}) with $\mathP\{F_1 (X_1) >u,\dots, F_d (X_d) >u\}$.
As a result, the pair of coefficients $(\chi, \eta)$ is able to measure the tail dependence strength across asymptotic dependence and independence classes, and it is often used for assessing model fit.

\subsection{Tail Dependence of the GH Distribution}\label{section:GHtail}
Recall that the GH distribution is elliptical if and only if $\bm{\gamma} = \bm{0}$. 
\cite{SF12} have proved that the elliptical GH distribution with dispersion matrix $\Sigma = \bigl(\begin{smallmatrix} 1 & \rho \\ \rho & 1 \end{smallmatrix}\bigr), 0\leq \rho < 1$, excluding the limiting case $\psi=0$, is asymptotically independent and has residual tail dependence coefficient $\eta=\sqrt{(1+\rho)/2}$; see also \cite{HOT}.
\cite{Nolde14} considered a more general class and claimed that the whole GH distribution family with dispersion matrix $\Sigma = \bigl(\begin{smallmatrix} 1 & \rho \\ \rho & 1 \end{smallmatrix}\bigr), |\rho| \neq 1$, has residual tail dependence coefficient $\sqrt{(1+\rho)/2}$, which implies that the skewness parameter $\bm{\gamma}$ has no influence on the residual tail dependence coefficient.
However, \cite{Hammerstein16} used a different approach and claimed that the GH distribution can be asymptotically dependent or independent depending on the choice of the parameters, although without giving the residual tail dependence coefficient when asymptotic independence is present.

After carefully investigating where this contradiction comes from, we found that both of their proofs have mistakes which lead to incorrect conclusions.
We detail their mistakes and a corrected version of von Hammerstein's proof in the Supplementary Material.
We now give a corrected description of the tail dependence property of the GH distribution, using a similar approach to \cite{Nolde14}.
In contrast with the often complicated calculation of the residual tail dependence coefficient by its definition, \cite{Nolde14} provided a geometric interpretation of this coefficient, which leads to simple and intuitive computations of this coefficient for a variety of distributions.
This is particularly the case when joint densities are easier to compute than joint distribution or survival functions.
We first recall some definitions and theorems from \cite{Nolde14} which will be used in our proof.

Consider a sequence of independent and identically distributed random vectors $\bm{Z}_1,\bm{Z}_2,\dots$ on $\mathR^2$. 
Let $N_n:=\{\bm{Z}_1/r_n,\dots,\bm{Z}_n/r_n\}$ denote an $n$-point sample cloud with scaling constants $r_n>0$, $r_n\rightarrow\infty$ as $n\rightarrow\infty$ and let $N_n(A)$ be the number of points of $N_n$ contained in the Borel set $A\in \mathR^d$, i.e. $N_n(A)=\sum_{i=1}^n \mathrm{1}_A(\bm{Z}_i/r_n)$.
Let $D$ be a compact set in $\mathR^2$. 
Then $D$ is called a \textit{limit set} of the sample cloud $N_n$ as $n\rightarrow\infty$ if (i) $\mathP\{N_n(U^c)>0\}\rightarrow 0$ for open sets $U$ containing $D$, and $U^c$ denotes the complement of $U$, and (ii) $\mathP\{N_n(\bm{p}+\varepsilon B)>m\} \rightarrow 1$ for all $m\geq 1, \varepsilon >0, \bm{p}\in D$, where $B$ denotes the Euclidean unit ball. 
If every ray from the origin intersects the boundary of a given set $D$ in a single point, then $D$ can be characterized by a continuous \textit{gauge function} $n_D$.
That is, $n_D: \mathR^2 \rightarrow [0,\infty)$ is a homogeneous function of degree one, i.e. $n_D(t\bm{x})=tn_D(\bm{x})$ for all $t>0$, and $D=\{\bm{x}\in\mathR^2: n_D(\bm{x})<1 \}$.
A set $D$ is called \textit{star shaped} if $\bm{x}\in D$ implies $t \bm{x}\in D$ for all $t\in (0,1)$.

A measurable function $\Psi: \mathR_+ \rightarrow \mathR_+$ is \textit{regularly varying} at infinity with exponent $a\in \mathR$ if for $x>0$,
$
\lim_{t\rightarrow\infty} \Psi(tx)/ \Psi(x) = x^a.
$
Roughly speaking, regularly varying functions behave asymptotically like power functions.
When $a=0$, the function $\Psi$ is slowly varying.

\begin{lemma}[Theorem 2.1 in \cite{Nolde14}]\label{Nolde1}
Let $\{(X_i,Y_i)^\top,i\geq 1 \}$ be independent random vectors in $\mathR^2$ from a distribution $G$ with marginal survival functions $1-G_i(s)=e^{-\Psi_i(s)},s>0$, where $\Psi_1$ and $\Psi_2$ are regularly varying at infinity with the same exponent $a>0$. Suppose there exists a scaling sequence $r_n >0$ such that the $n$-point sample cloud 
\[
N_n = \{(X_1/r_n,Y_1/r_n)^\top,\dots,(X_n/r_n,Y_n/r_n)^\top \}
\]
converges onto set $D$ whose interior is bounded, open, and star-shaped as $n\rightarrow \infty$. 
Let $\bm{q}=(q_1,q_2)^\top=\sup D$ denote the coordinate-wise supremum of $D$. Then $D$ is a star-shaped subset of $(-\infty,\bm{q}]$. Define 
\[
r_D = \min\{r\geq0: D\cap((rq_1,\infty)\times(rq_2,\infty))=\emptyset \}.
\]
If $\bm{q}\notin D$, then $r_D <1$, $X_1$ and $Y_1$ are asymptotically independent and the residual dependence coefficient is $\eta = r_D^a$.
\end{lemma}

The definition of $r_D$ implies that, if $D$ is convex, $r_D \bm{q}$ lies on the boundary of $D$ and thus $n_D(r_D \bm{q})=1$.
Using the homogeneity of $n_D$ and Lemma \ref{Nolde1} one can get $\eta=n_D(\bm{q})^{-a}$; see also \cite{Nolde2020} for further details.
In practice, a random vector is often described by its multivariate probability density function.
The relation between the density function and its limit set is given in the following lemma and thereby a convenient way to calculate the residual tail dependence coefficient $\eta$ is given.

\begin{lemma}[Proposition 3.1 in \cite{Nolde14}]\label{Nodel2}
Let $\bm{Z}_1,\bm{Z}_2,\dots$ be independent and identically distributed random vectors with a continuous positive density $g$ on $\mathR^d$. Let $\gamma:= -\log g$. Suppose there exists a function $\ell$ on $\mathR^d$ which is positive outside a bounded set, and a nonzero vector $\bm{v}$ such that
\[
\frac{\gamma(r_n\bm{u}_n)}{\gamma(r_n\bm{v})} \rightarrow \ell(\bm{u}), \text{ for some } r_n \rightarrow\infty, \text{ and any } \bm{u}_n \rightarrow \bm{u}\in \mathR^d.
\]
Then the sequence of sample clouds $N_n=\{\bm{Z}_1/r_n,\dots,\bm{Z}_n/r_n\}$ converges as $n\rightarrow\infty$ onto a limit set $D$ whose gauge function is given by $n_D=\ell^\theta$ for some positive constant $\theta$.
\end{lemma}

For details about the proof of Lemma \ref{Nolde1} and Lemma \ref{Nodel2} and further background knowledge, we refer to \cite{Nolde14} and references therein.
We are now ready to use these two results to show that the GH distribution, except in the limiting case $\psi=0$, is asymptotically independent and we give its residual dependence coefficient.

\begin{proposition}\label{mainresult}
Let $(X_1, X_2)^\top \sim \text{GH}_2(\lambda,\kappa,\psi,\bm{\gamma},\bm{\mu},\Sigma)$ with $\bm{\gamma}=(\gamma_1,\gamma_2)^\top$, $\bm{\mu}=(\mu_1,\mu_2)^\top$, and $\Sigma = \bigl(\begin{smallmatrix} 1 & \rho \\ \rho & 1 \end{smallmatrix}\bigr), |\rho| \neq 1$. If $\psi > 0$, then $X_1, X_2$ are asymptotically independent and the residual dependence coefficient is
     \[
     \eta = \frac{1-\rho^2}{\sqrt{\{\psi(1-\rho^2)+\gamma_1^2-2\rho\gamma_1\gamma_2+ \gamma_2^2\}(m_1^2-2\rho m_1 m_2+ m_2^2)}-m_1(\gamma_1-\rho\gamma_2)-m_2(\gamma_2-\rho\gamma_1)},
     \]
     where
     \begin{equation*}
         m_1 = \frac{\gamma_1+\sqrt{\psi+\gamma_1^2}}{\psi}, \quad
         m_2 = \frac{\gamma_2+\sqrt{\psi+\gamma_2^2}}{\psi}.
     \end{equation*} 
\end{proposition}
\begin{proof}
Note that both the tail dependence coefficient and the residual dependence coefficient are copula properties, i.e., they are invariant under strictly increasing marginal transformations. 
Hence, for simplicity we assume $\bm{\mu}=\bm{0}$ and study the tail dependence of $(X_1,X_2)^\top \sim \text{GH}_2(\lambda,\kappa,\psi,\bm{\gamma},\bm{0},\Sigma)$.
Let $f$ be the probability density function of $(X_1,X_2)^\top$.
Using the asymptotic property of the Bessel function (see formula (9.7.2) in \cite{Abramowitz1972}), we know that $K_{\lambda}(x) = \sqrt{\frac{\pi}{2x}}\exp(-x) + o(1/\sqrt{x})$ as $x\rightarrow\infty$.
Then, for $\bm{u}, \bm{v} \in \mathR^2$ and $s\rightarrow\infty$, we have 
\begin{align*}
    \frac{-\log\{f(s\bm{u})\}}{-\log\{f(s\bm{v})\}} &\sim \frac{\sqrt{(\psi+\bm{\gamma}^\top\Sigma^{-1}\bm{\gamma})(\kappa+s^2\bm{u}^\top\Sigma^{-1}\bm{u})}-s\bm{u}^\top\Sigma^{-1}\bm{\gamma}} {\sqrt{(\psi+\bm{\gamma}^\top\Sigma^{-1}\bm{\gamma})(\kappa+s^2\bm{v}^\top\Sigma^{-1}\bm{v})}-s\bm{v}^\top\Sigma^{-1}\bm{\gamma}} \\
    &\sim \frac{\sqrt{(\psi+\bm{\gamma}^\top\Sigma^{-1}\bm{\gamma})\bm{u}^\top\Sigma^{-1}\bm{u}}-\bm{u}^\top\Sigma^{-1}\bm{\gamma}} {\sqrt{(\psi+\bm{\gamma}^\top\Sigma^{-1}\bm{\gamma})\bm{v}^\top\Sigma^{-1}\bm{v}}-\bm{v}^\top\Sigma^{-1}\bm{\gamma}}.
\end{align*}
For fixed $\bm{v}\neq \bm{0}$, let $h:=\sqrt{(\psi+\bm{\gamma}^\top\Sigma^{-1}\bm{\gamma})\bm{v}^\top\Sigma^{-1}\bm{v}}-\bm{v}^\top\Sigma^{-1}\bm{\gamma}$.
Then, consider the set $D$:
\[
D = \{\bm{u}\in\mathR^2: \big(\sqrt{(\psi+\bm{\gamma}^\top\Sigma^{-1}\bm{\gamma})\bm{u}^\top\Sigma^{-1}\bm{u}}-\bm{u}^\top\Sigma^{-1}\bm{\gamma}\big)/h \leq 1\}.
\]
Note that $\big(\sqrt{(\psi+\bm{\gamma}^\top\Sigma^{-1}\bm{\gamma})\bm{u}^\top\Sigma^{-1}\bm{u}}-\bm{u}^\top\Sigma^{-1}\bm{\gamma}\big)/h \leq 1$ is equivalent to
\[
(\psi+\bm{\gamma}^\top\Sigma^{-1}\bm{\gamma})\bm{u}^\top\Sigma^{-1}\bm{u} \leq (h+\bm{u}^\top\Sigma^{-1}\bm{\gamma})^2,
\]
which can be simplified to
\begin{equation}\label{equa}
(\psi+\gamma_2^2)u_1^2 + (\psi+\gamma_1^2)u_2^2 - 2(\rho\psi+\gamma_1\gamma_2)u_1u_2 - 2h\{(\gamma_1-\rho\gamma_2)u_1+(\gamma_2-\rho\gamma_1)u_2\} \leq h^2(1-\rho^2).
\end{equation}
Since $\psi > 0, \rho\neq 1$, we have 
\begin{align*}
    \Delta &= 4(\rho\psi+\gamma_1\gamma_2)^2 - 4(\psi+\gamma_2^2)(\psi+\gamma_1^2) \\
    &= -4(1-\rho^2)\psi^2-4\psi(\gamma_1^2+\gamma_2^2-2\rho\gamma_1\gamma_2) \\
    &<0.
\end{align*}
This implies that $D$ is an ellipse, and thus convex.
By Lemma \ref{Nodel2}, $D$ is the limit set and its associated gauge function is $n_D(\bm{u})=\{\sqrt{(\psi+\bm{\gamma}^\top\Sigma^{-1}\bm{\gamma})\bm{u}^\top\Sigma^{-1}\bm{u}}-\bm{u}^\top\Sigma^{-1}\bm{\gamma}\}/h$.
Furthermore, note that $n_D(\bm{0}) \leq 1$, which implies that $\bm{0}$ is contained in $D$.
Hence, $D$ is bounded and star shaped.
Now it remains to find the coordinatewise supremum of $D$.

Note that the inequality (\ref{equa}) can be written as
\[ 
(\psi+\gamma_2^2)u_1^2 - 2[(\rho\psi+\gamma_1\gamma_2)u_2+h(\gamma_1-\rho\gamma_2)]u_1 + (\psi+\gamma_1^2)u_2^2 - 2h(\gamma_2-\rho\gamma_1)u_2 \leq h^2(1-\rho^2).
\]
Denote by $\bm{q}=(u_1^*, u_2^*)^\top$ the coordinatewise supremum of $D$, then we have
\[
(\psi+\gamma_1^2)u_2^{*2} - 2h(\gamma_2-\rho\gamma_1)u_2^* = h^2(1-\rho^2) + \frac{\{(\rho\psi+\gamma_1\gamma_2)u_2^*+h(\gamma_1-\rho\gamma_2)\}^2}{\psi+\gamma_2^2},
\]
which can be simplified to
\[
\psi u_2^{*2} - 2h\gamma_2 u_2^* -h^2 = 0.
\]
Hence, 
\[
u_2^* = \max\Big(\frac{2h\gamma_2+\sqrt{4h^2\gamma_2^2+4h^2\psi}}{2\psi}, \frac{2h\gamma_2-\sqrt{4h^2\gamma_2^2+4h^2\psi}}{2\psi} \Big) = \frac{h(\gamma_2+\sqrt{h^2\gamma_2^2+h^2\psi})}{\psi}.
\]
Similarly, we get $u_1^*=h(\gamma_1+\sqrt{h^2\gamma_1^2+h^2\psi})/\psi$.
Since $(X_1,X_2)^\top \sim \text{GH}_2(\lambda,\kappa,\psi,\bm{\gamma},\bm{0},\Sigma)$, we know that $X_i \sim \text{GH}_2(\lambda,\kappa,\psi,\gamma_i,0,1), i=1,2$. 
Denote the marginal density function and distribution function of $X_i ,i=1,2$ as $f_{X_i}$ and $F_{X_i}$, respectively, and write $1-F_{X_i}(x)=e^{-\Psi_i (x)}$. 
Using the asymptotic relation $K_{\lambda}(x) \sim \sqrt{\pi/(2x)}e^{-x}, x\rightarrow\infty$ we get
\[
f_{X_i} (x) \sim c_i x^{\lambda-1} \exp\big\{-\big(\sqrt{\psi+\gamma_i^2}-\gamma_i\big) x \big\} \text{ as } x\rightarrow\infty,
\]
where $c_i$ is a constant.
As $\psi >0$ we know that $\sqrt{\psi+\gamma_i^2}-\gamma_i > 0$.
By Proposition 2 in \cite{Hammerstein16}, we have 
\[
1-F_{X_i}(x) \sim \frac{c_i x^{\lambda-1}}{\sqrt{\psi+\gamma_i^2}-\gamma_i}  \exp\big\{-\big(\sqrt{\psi+\gamma_i^2}-\gamma_i\big) x \big\} \text{ as } x\rightarrow\infty.
\]
Hence, the functions $\Psi_i, i=1,2$ are both regularly varying at infinity with the same exponent $a=1$.

As the supremum point $\bm{q}$ of an ellipse $D$ satisfies $\bm{q}\notin D$, by Lemma \ref{Nolde1} and Lemma \ref{Nodel2}, $X_1$ and $X_2$ are asymptotically independent and the residual dependence coefficient is
\begin{align*}
    \eta &= n_D(\bm{q})^{-a} \\
         &= [\{\sqrt{(\psi+\bm{\gamma}^\top\Sigma^{-1}\bm{\gamma})\bm{q}^\top\Sigma^{-1}\bm{q}}-\bm{q}^\top\Sigma^{-1}\bm{\gamma}\}/h ]^{-1} \\
         &= \frac{1-\rho^2}{\sqrt{\{\psi(1-\rho^2)+\gamma_1^2-2\rho\gamma_1\gamma_2+ \gamma_2^2\}(m_1^2-2\rho m_1 m_2+ m_2^2)}-m_1(\gamma_1-\rho\gamma_2)-m_2(\gamma_2-\rho\gamma_1)},
\end{align*}
where
\begin{equation*}
         m_1 = \frac{\gamma_1+\sqrt{\psi+\gamma_1^2}}{\psi}, \quad
         m_2 = \frac{\gamma_2+\sqrt{\psi+\gamma_2^2}}{\psi}.
\end{equation*} 
\end{proof}

When $\rho=0$, which corresponds to independence between the components of the Gaussian random vector $\bm{W}$ and implies that the dependence between $X_1$ and $X_2$ is fully specified by the mixing variable $R$, we have $\eta=\{(\psi+\gamma_1^2+\gamma_2^2)^{1/2}(m_1^2+m_2^2)^{1/2}-m_1\gamma_1-m_2\gamma_2\}^{-1}$, where $m_1, m_2$ are the same as above.
One further observation is that when $\psi\rightarrow\infty$, which means the mixing variable $R$ has a very light tail, we have $\eta=\sqrt{(1+\rho)/2}$.
In this case, we obtain the same $\eta$ as the elliptical GH distribution, i.e., when $\bm{\gamma}=\bm{0}$, and it equals the square root of the residual dependence coefficient of a bivariate Gaussian random vector with correlation $\rho$.

It is also important to note that when $\psi=0, \bm{\gamma}\neq\bm{0}$, the sample cloud does not converge onto a bounded set and thus the method we used above fails.
When $\psi=0, \bm{\gamma}=\bm{0}$, the GH distribution reduces to the (asymptotically dependent) Student-$t$ distribution and the tail dependence coefficient is known from the literature; see \cite{Embrechts2001}.

\begin{figure}[!t]
    \centering
    \includegraphics[scale=0.6]{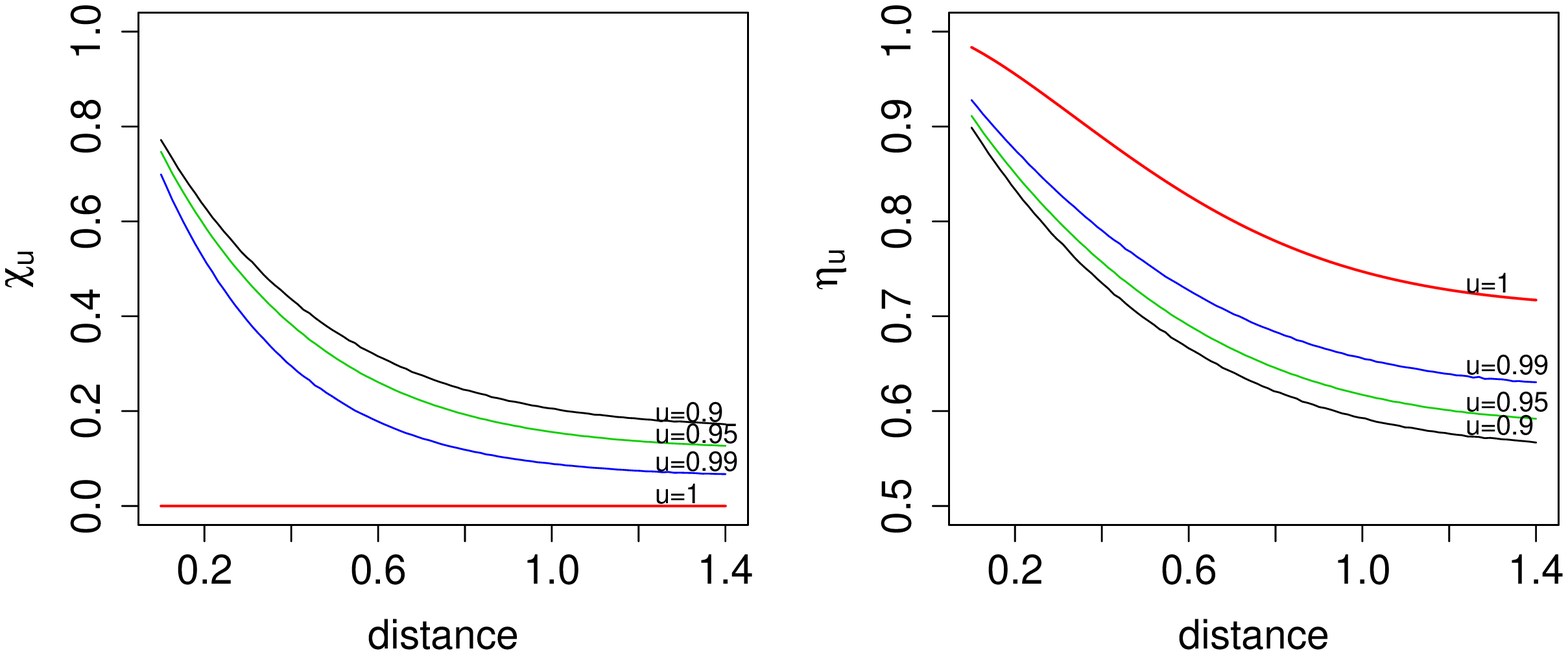}
    \caption{Bivariate $\chi_u$, $\eta_u$ for $u=0.9, 0.95, 0.99$ and their limits as $u\rightarrow 1$ of the GH distribution for $\lambda=-0.5, \kappa=\psi=1, \bm{\gamma}=\bm{0}$, and $\Sigma$ being a correlation matrix determined by a powered exponential correlation function $\rho(\bm{s}_1,\bm{s}_2)=\exp\{-(\|\bm{s}_1 -\bm{s}_2\|/0.6)^{3/2}\}$, where $\|\bm{s}_1 -\bm{s}_2\|$ is the Euclidean distance between sites $\bm{s}_1$ and $\bm{s}_2$. For $u<1$, $\chi_u$ and $\eta_u$ are calculated numerically using their definition (\ref{chidefinition}) and (\ref{etadefinition2}), and their limits are calculated using the results in Proposition \ref{mainresult}.}
    \label{figure::eta}
\end{figure}

Figure \ref{figure::eta} illustrates the flexibility in extremal dependence structures of the GH distribution, by plotting the bivariate $\chi_u$, $\eta_u$ and their limits as a function of distance between the pairs, for a range of values of $u\in[0.9,1)$, and $\lambda=-0.5, \kappa=\psi=1, \bm{\gamma}=\bm{0}$, $\Sigma$ being a correlation matrix determined by a powered exponential correlation function $\rho(\bm{s}_1,\bm{s}_2)=\exp\{-(\|\bm{s}_1 -\bm{s}_2\|/0.6)^{3/2}\}$, where $\|\bm{s}_1 -\bm{s}_2\|$ is the Euclidean distance between sites $\bm{s}_1$ and $\bm{s}_2$.
The results indicate the slow convergence of $\chi_u$ and $\eta_u$ to their limits, and at any observable levels, including $u=0.99$, nonnegligible dependence may exist ($\chi_u$ could be as large as 0.4 or 0.5 for pairs at short distances), but the pairs are asymptotically independent.
This emphasizes again the need for asymptotically independent models or models that can capture both asymptotic independence and dependence.

\section{Copula-Based Inference}
\subsection{Copula Model and Likelihood Inference}
As here we focus on modeling the spatial dependence, we take a copula approach.
A copula is a multivariate distribution function with standard uniform margins.
Thanks to Sklar's theorem \citep{Sklar1959}, for any continuous multi-dimensional distribution function there is a unique copula associated with it.
This is the case for the multivariate GH distribution and we call the associated copula the GH copula.
Specifically, suppose $(X_1,\dots,X_d)^\top$ follow a $d$-dimensional GH distribution $F$, which implies that their marginal distribution functions $F_1,\dots,F_d$ are univariate GH distributions, then the GH copula is defined by
\[
C(\bm{u}) = \mathP\{F_1(X_1)\leq u_1,\dots,F_d(X_d)\leq u_d\} = F\{F_1^{-1}(u_1),\dots,F_d^{-1}(u_d)\}, \quad \bm{u}\in [0,1]^d,
\]
where $F_i^{-1}, i=1,\dots,d$, is the inverse of $F_i$.
Its density function thus can be derived easily as
\[
c(\bm{u}) = \frac{\partial^d}{\partial u_1 \cdots \partial u_d} C(\bm{u}) = \frac{f\{F_1^{-1}(u_1),\dots,F_d^{-1}(u_d)\}}{f_1\{F_1^{-1}(u_1)\}\cdots f_d\{F_d^{-1}(u_d)\}},
\]
where $f$ is the joint density function and $f_1,\dots,f_d$ are the marginal GH density functions of $X_1,\dots,X_d$, respectively.

Due to the closed-form density of the multivariate GH distribution, inference based on the full likelihood is feasible.
Here both in the simulation study and data applications, we fit the GH copula model to the whole dataset, rather than using a censored likelihood approach as \cite{Wadsworth2012}.
By doing so, we gain two major benefits: (i) we can avoid the computational issues associated with the censored likelihood, which requires calculations of potentially high-dimensional distribution functions and that often have to be evaluated numerically; (ii) we can model dependence both in the bulk and in the tail.
However, we also point out that an obvious disadvantage of taking such a full likelihood approach is that the extremal dependence structure would not be as well captured as using a censored likelihood approach.
Here we choose the full likelihood approach because the GH copula model is highly flexible, more so than most other spatial models.
This means that it could be a reasonable approach, and the simulation study will support this.

We leverage the R \citep{R2020} package \textit{Rcpp} to evaluate the GH density function with low-level language C++ and the package \textit{ghyp} to compute the quantiles of the GH distribution, which uses splines to interpolate the univariate GH distribution function and then finds the root efficiently with the uniroot function in R.
These ``tricks'' have improved the efficiency of the algorithm and, as an illustration, the computational time for fitting the GH copula to a dataset with 100 sites and 2000 observations at each site is less than three hours on a laptop with 8 GB memory and 2.3 GHz Intel Core i5 processor.

\subsection{Simulation Study}
Simulation from the multivariate GH distribution is straightforward due to its stochastic representation (\ref{NMVM_repre}).
Specifically, independently sampling $R\sim\text{GIG}(\lambda,\kappa,\psi)$ and $\bm{W}\sim\mathcal{N}_d(\bm{0},\Sigma)$ and plugging them into formula (\ref{NMVM_repre}) yields exact samples from the distribution $\text{GH}_d(\lambda,\kappa,\psi,\bm{\gamma},\bm{\mu},\Sigma)$.
Sampling from a multivariate normal distribution is easy and sampling from a generalized inverse Gaussian distribution is feasible using the algorithm in \cite{Dagpunar1989}, which is implemented in the R package \textit{ghyp}.
Once a sample is generated from the multivariate GH distribution, one can obtain a sample for the GH copula by transforming the marginals to uniform distribution, using the probability integral transform.

To reduce the number of parameters in the GH copula model, in both the simulation study and data application, we here assume that $\bm{\gamma}=\gamma_c \bm{1}_d, \gamma_c\in\mathR$ and $\Sigma$ is a correlation matrix determined by a powered exponential correlation function $\rho(\bm{s}_1,\bm{s}_2)=\exp\{-(\|\bm{s}_1 -\bm{s}_2\|/\zeta)^\nu\}$, where $\|\bm{s}_1 -\bm{s}_2\|$ is the Euclidean distance between sites $\bm{s}_1$ and $\bm{s}_2$, $\zeta>0$ is the range parameter, and $\nu\in(0,2]$ controls the smoothness of the realized random fields. 
In the first simulation study we consider 25 sites on the uniform grid $\{0, 0.25, 0.5, 0.75, 1\}^2$, and set $\lambda=\kappa=\psi=\gamma_c=1$, $\zeta=0.2 \text{ or } 0.7$ (short or long range dependence), $\nu=0.5 \text{ or } 1.5$ (rough or smooth random field).
Then we generate $n=250, 500 \text{ or } 1000$ replicates from the GH copula at these sites and fit the GH copula to the simulated dataset of each scenario.

\begin{figure}[!t]
    \centering
    \includegraphics[scale=0.6]{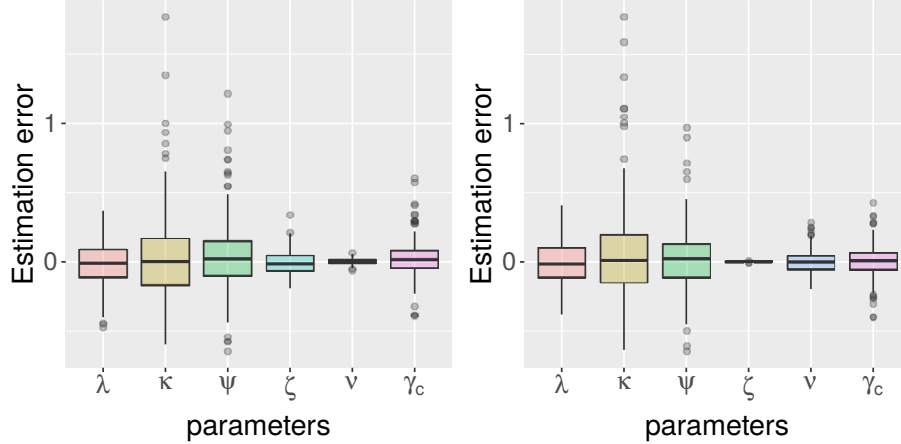}
    \caption{Boxplots of estimation errors of all the parameters for different scenarios in the first simulation study. Left panel: scenario $\zeta=0.7, \nu=0.5, n=500$; right panel: scenario $\zeta=0.2, \nu=1.5, n=500$.}
    \label{figure::scenario}
\end{figure}

We run the experiment for 300 times and Figure \ref{figure::scenario} depicts the boxplots of estimation errors (i.e., differences between parameter estimates and the true parameters) for all of the six parameters when $\zeta=0.7, \nu=0.5, n=500$ and $\zeta=0.2, \nu=1.5, n=500$.
The results indicate that the spatial range parameter $\zeta$ and smoothness parameter $\nu$ are relatively easier to estimate than the parameters $\lambda,\kappa,\psi$, which determine the shape of the mixing GIG distribution, and the skewness parameter $\gamma_c$.
We also observe that changing the values of $\zeta$ and $\nu$ does not affect the estimation of other parameters, but the larger value of $\zeta$, which corresponds to stronger spatial dependence between sites, leads to smaller variability of the estimate of $\nu$, and larger value of $\nu$ also yields smaller variability of the estimate of $\zeta$.
One further observation is that when we increase the number of replicates $n$, the variability of the estimates of all parameters decreases, roughly at a rate of $\sqrt{n}$ as expected, as in shown in Figure \ref{figure::repeff}.

\begin{figure}[!t]
    \centering
    \includegraphics[scale=0.6]{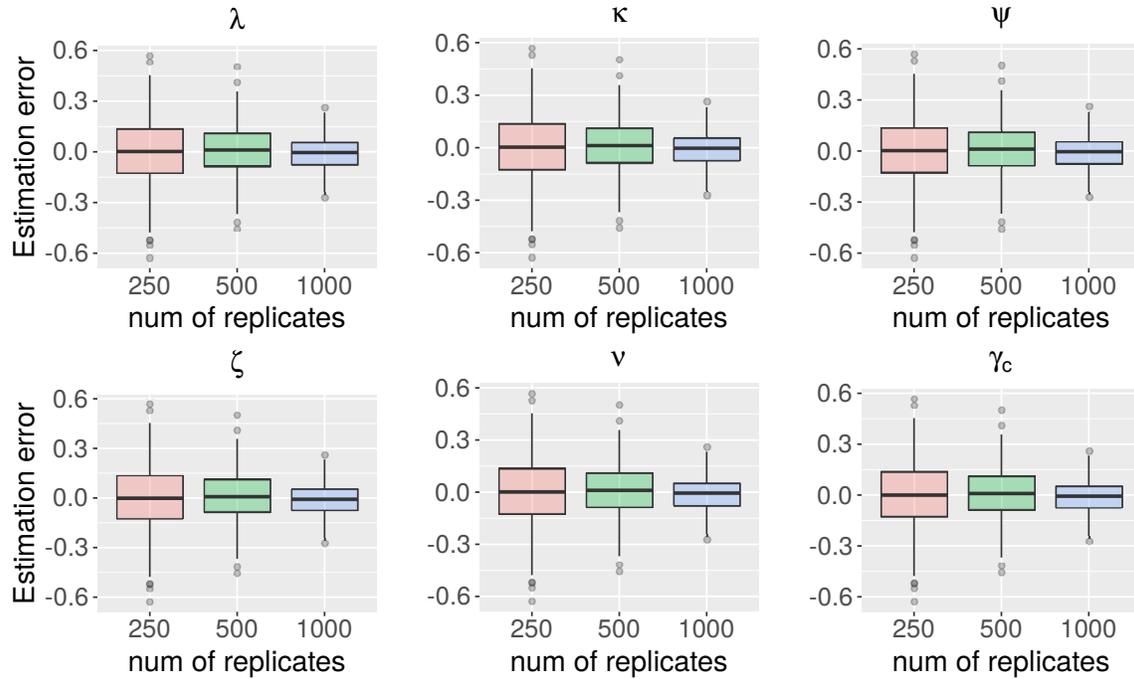}
    \caption{Boxplots of estimation errors of all the parameters for the scenario $\zeta=0.7, \nu=1.5$ with different number of replicates $n$ in the first simulation study.}
    \label{figure::repeff}
\end{figure}

In the second simulation study we consider 100 sites on the grid $\{0,1/9,2/9,\dots,8/9,1\}^2$ and generate 300 replicates from the inverted Brown--Resnick process \citep{Kabluchko2009,Wadsworth2012} using the extremal functions approach of \cite{Dombry2016}, with the variogram of the Brown--Resnick process defined as $\gamma(\bm{h})=2\|\bm{h}\|/\zeta$, where $\zeta=0.3$.
This model is known to be asymptotically independent with bivariate residual tail dependence coefficient function $\eta(\bm{h})=1/[2\Phi\{\sqrt{\gamma(\bm{h})}/2\}]$.
To illustrate the flexibility of the GH distribution in a misspecified setting, we then fit the GH copula and the normal inverse Gaussian (NIG) copula, which is a subclass of the GH copula when $\lambda=-1/2$, to the whole simulated dataset.
Figure \ref{figure::invBR} displays the true residual tail dependence coefficient of the inverted Brown--Resnick process, its counterpart of the fitted GH copula and NIG copula model, and empirical estimates of $\eta_u$, as defined in (\ref{etadefinition2}), based on the simulated dataset.
The results show that, despite fitting the GH copula to the entire dataset without censoring low non-extreme values, the GH copula captures the residual tail dependence quite well for the simulated dataset, and the shape of the fitted residual tail dependence coefficient function, with respect to distance, resembles its true function very closely, although a slight bias exists as expected.

\begin{figure}[!t]
    \centering
    \includegraphics[scale=0.55]{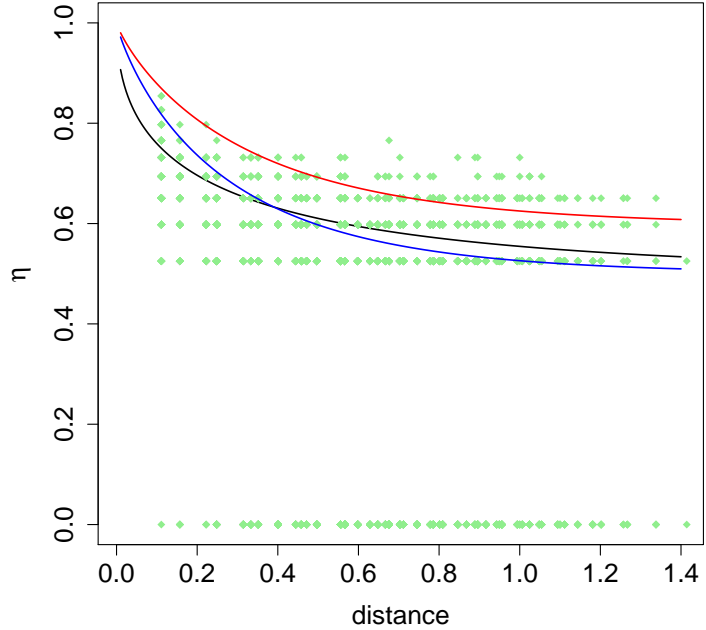}
    \caption{Estimated bivariate residual tail dependence coefficients in the second simulation study. Black line: true $\eta$ of the inverted Brown--Resnick process; red line: theoretical $\eta$ of the fitted normal inverse Gaussian copula, calculated using Proposition \ref{mainresult}; blue line: theoretical $\eta$ of the fitted GH copula, calculated using Proposition \ref{mainresult}; green diamond: empirical tail dependence coefficient $\eta_u$ with $u=0.95$.}
    \label{figure::invBR}
\end{figure}

\section{Environmental Applications}
\subsection{Wave Height Example}
We first consider the hindcast dataset of significant wave heights analyzed by \cite{Wadsworth2012} and \cite{Huser2019}.
This dataset contains eight observations per day over a period of 31 years, at 50 spatial locations in the North Sea.
To reduce temporal dependence and ease the computational burden, only one observation per day in the winter months (December, January, and February) at a subset of 20 sites has been used to fit the Huser--Wadsworth model from \cite{Huser2019}.
Here, we consider the same setting, although the GH copula or its subclasses are computationally more efficient to fit due to the closed-form density of the GH distribution.

The margins at each site are transformed to standard uniform following the semiparametric approach of \cite{ColesTawn1991}.
Then, three different models are fitted to the standardized dataset, namely the Huser--Wadsworth model, the Gaussian copula model and the GH copula model.
To illustrate the flexibility of our proposed GH model to capture the upper tail dependence in a disadvantageous situation, we first use a censored likelihood scheme to fit the Huser--Wadsworth model and the Gaussian copula model, which helps them better capture the extremal dependence structure; see equation (18) in \cite{Huser2019} for the censoring scheme.
Then we fit the GH copula to the data without censoring low observations.
Figure \ref{figure::waveheight} displays the fitted values of 20-variate $\chi_u$ and $\eta_u$, as defined in Section \ref{taildependence}, for the three models.
The uncertainty measures are based on the same stationary bootstrap procedure as in \cite{Huser2019}, i.e., 200 bootstrap samples are generated using the stationary bootstrap which sample blocks of geometric length.
Strikingly, the results show that the GH copula can capture the extremal dependence better than the Gaussian copula, even though a censored likelihood has been used for the Gaussian copula.
The performance of the GH copula is slightly worse than the Huser--Wadsworth copula, but the differences are fairly minor overall and the GH model still provides a very good fit in the upper tail despite its full uncensored likelihood estimation approach.
We stress again that the Huser--Wadsworth copula is specially designed to model the upper tail and a censored likelihood has been used for this purpose, while our proposed GH model is here used to model the entire distribution, from low to high quantiles.
We also fit the Huser--Wadsworth model and the Gaussian model without censoring the low observations.
The log-likelihood of GH copula is around 2000 units larger than that of the Huser--Wadsworth model and 6000 units larger than that of the Gaussian copula model, which indicates that the GH copula is clearly more flexible than the other two if all observations are considered.

\begin{figure}[!t]
    \centering
    \includegraphics[scale=0.6]{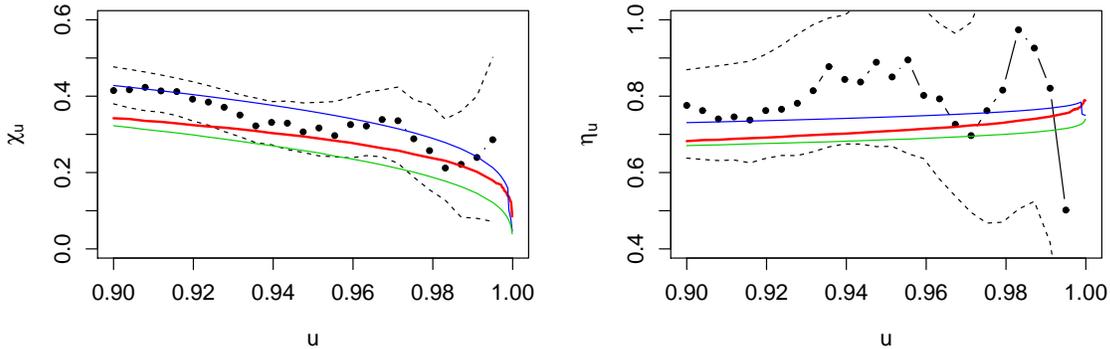}
    \caption{Estimates of $20$-variate $\chi_u$ (left panel) and $\eta_u$ (right panel) for the hindcast wave height data. Central black dots: empirical estimate of $\chi_u$ and $\eta_u$; dashed line: approximate $95\%$ confidence intervals based on a stationary bootstrap procedure; thick solid red line: fit from the GH copula; thin solid blue line: fit from the Huser-Wadsworth model; thin solid green line: fit from the Gaussian model.}
    \label{figure::waveheight}
\end{figure}

\subsection{Wind Gust Example}
To illustrate the computational benefits of our proposed model, we now consider a higher-dimensional dataset of daily maximum wind gusts from the state of Oklahoma, USA, which can be freely downloaded from mesonet.org.
This dataset contains daily measurement of highest 3-second wind speed from January 1, 1995 to December 31, 2020, at 120 observational stations across the state of Oklahoma.
To induce approximate stationarity, we only consider the summer observations, i.e., observations in July, August and September, which are often the highest in a year.
As some observational stations have many missing values, we delete the ones with more than 100 missing values, resulting in 95 sites with 2392 observations at each site.
The minimum distance between sites is around 12 km and the maximum distance is around 800 km.
The time series at each site appear approximately stationary, hence we work directly with this dataset and transform all the marginals to standard uniform using the nonparametric approach based on ranks.
Alternatively, we could use the semiparametric approach of \cite{ColesTawn1991} as in the wave height example.

We then fit the GH copula model to the standardized dataset, as well as four of its subclasses and limiting models, namely the Gaussian copula, the $t$ copula model, the hyperbolic copula and the NIG copula model, based on a full likelihood approach.
The computational time for fitting the GH copula is less than three hours on a laptop, which is a significant speed-up compared to the censored Huser--Wadsworth model in such high dimensions. 
It takes less than ten minutes to fit the Gaussian copula and $t$ copula due to the efficient computation of their quantile functions in R.
Table \ref{tab:model_AIC} shows the number of parameters in each model and the Akaike information criterion (AIC) of the fitted models.
The results indicate that the general GH copula has the best performance as expected, and the NIG submodel has the second best performance.
The large difference between the AIC values of the GH copula model and its limiting cases, the widely used Gaussian copula model and the $t$ copula model, implies that there is a clear advantage of using the GH copula model for this dataset despite its higher model complexity.

\begin{table}[!t]
\centering
\caption{Number of parameters and AIC values for the different models fitted in our wind gust data application}
\begin{tabular}{ccc}\toprule
Model & $\#$ of parameters  &  AIC \\ \midrule
Gaussian copula & 2 & $-298898$  \\
$t$ copula & 3 & $-324069$  \\
Hyperbolic copula & 5 & $-316636$  \\
NIG copula & 5 & $-327035$  \\
GH copula & 6 & $\bm{-328141}$  \\
\bottomrule
\end{tabular}
\label{tab:model_AIC}
\end{table}

To assess how well our model captures the dependence both in the bulk and the tail of the data, we consider Spearman's rank correlation $S_\rho$ and the bivariate extremal dependence measures $\chi_u$ and $\eta_u$ described in Section \ref{taildependence}.
We first calculate the empirical estimates of $S_\rho$, $\chi_u$ and $\eta_u$ with $u=0.95$ from the data, and their model-based counterparts from the estimated model.
Then we compute the bivariate kernel density estimators for pairs of observed empirical $S_\rho$ and fitted $S_\rho$, pairs of observed empirical $\chi_u$ and fitted $\chi_u$, and pairs of observed empirical $\eta_u$ and fitted $\eta_u$, using the R package \textit{ks}.
Figure \ref{fig:windgustfit} depicts the contours at level $25\%, 50\%$ and $75\%$ of these kernel density estimators.
The results show that our model can capture strong dependencies (at short and moderate distances) both in the bulk and the tail very well, but it overestimates weak dependencies (at large distances).
This is not very surprising as there are fewer pairs of sites that are weakly dependent (at long distances), so that strongly dependent pairs influence model estimation to a higher degree.
In particular, the spatially-constant mixing variable $R$ is constrained to provide a good fit at short distances, which simultaneously induces stronger dependencies across the entire spatial domain.
This can also be seen from the discussion on the values of $\eta$ for special cases of the GH distribution in Section \ref{section:GHtail}.
When the mixing variable $R$ has very light tails ($\psi\rightarrow\infty$), as the distance between sites tends to infinity and their correlation $\rho\rightarrow 0$, the residual dependence coefficient $\eta$ would converge to $1/\sqrt{2}$, which is much larger than $1/2$ and indicates that a strong dependence is still present in this case.

\begin{figure}[!t]
    \centering
    \includegraphics[width=\textwidth]{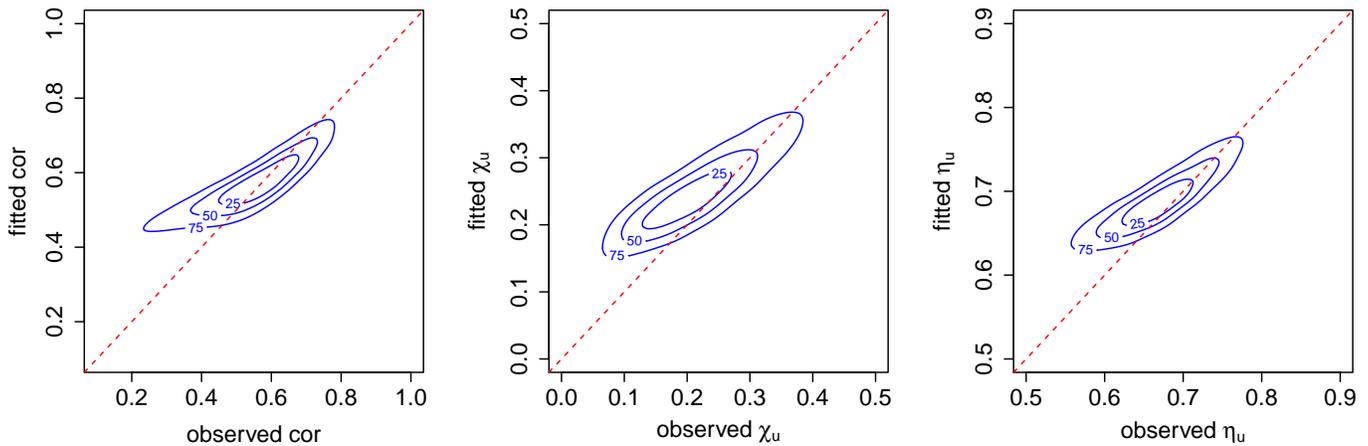}
    \caption{Contours at level $25\%, 50\%$ and $75\%$ of the bivariate kernel density estimators for pairs of observed empirical $S_\rho$ and fitted $S_\rho$ (left), pairs of observed empirical $\chi_u$ and fitted $\chi_u$ (middle), and pairs of observed empirical $\eta_u$ and fitted $\eta_u$ (right). A good fit should have density contours concentrated around the red dashed diagonal line.}
    \label{fig:windgustfit}
\end{figure}

\section{Discussion}
The GH distribution has a long history, and found popularity in financial modeling.
In this work we have re-investigated its tail dependence properties, pointed out the mistakes in its derivations in the literature, and gave a corrected description.
Based on this result, we propose to use the GH copula for spatial extremes when asymptotic independence is present, which contributes to the spatial extremes literature, as existing models in this setting that are both flexible and can be fitted efficiently are scarce.
We demonstrate the flexibility of this model both by a simulation study and two environmental data applications.

Unlike the asymptotic dependence case, where the scientific community broadly agrees that max-stable models should be fitted to block maxima, and generalized Pareto models should be fitted to threshold exceedances,
there is no such consensus in the asymptotic independence case yet.
\cite{Wadsworth2012} argue that it is often more natural, especially in the presence of asymptotic independence, to model the extremes of original events rather than site-wise maxima, and they suggest using a censored likelihood approach to better capture the extremal dependence.
Here, we did not censor the low observations when using the GH copula, because, as shown in the wave height data application, the non-censored GH copula can capture the tail dependence even better than the censored Gaussian copula and the high flexibility of the GH copula allows this full likelihood approach.
Moreover, in this case we can gain significant computational efficiency and also capture dependence in the bulk, as shown in the wind gust data application.

One interesting research direction is to investigate the tail properties in the limiting case of the GH distribution when $\psi=0$, which is an interesting distribution in its own as its marginal distribution has asymmetric tails, i.e., one has a power decay and the other has an exponential decay, and it reduces to the $t$ distribution when the skewness parameter $\bm{\gamma}$ tends to zero.
Unfortunately, the two different methods that we used in this work to derive the tail properties both fail in the limiting case when $\psi=0$ and the skewness parameter is not zero.
Hence, this remains an open problem.
Another future research direction is to investigate the tail dependence properties of more general normal mean-variance mixtures when the mixing distribution is specified only through its tail behavior.
As discussed in the wind gust data application, a weakness of the GH copula model is that independence cannot be captured as distance tends to infinity.
This is a typical weakness for all random scale (and location) models. 
Hence, one can try to overcome this issue and improve the performance of the GH copula model to better capture long-range weak dependence structures, for which the random partitioning approach of \cite{Morris2017} might be useful.

\section*{Acknowledgement}
This publication is based upon work supported by the King Abdullah Univer- sity of Science and Technology (KAUST) Office of Sponsored Research (OSR) under Award No. OSR-CRG2017-3434.
We gratefully acknowledge Philip Jonathan of Shell Research for providing the wave height data analysed in Section 4.1.



\bibliography{reference}

\end{document}